\documentclass[journal]{IEEEtran}
\usepackage{hyperref}
\usepackage{braket, physics, amssymb}
\usepackage{slashed}
\usepackage{tikz-feynman}
\usepackage{tikz}
\usepackage{comment}
\usetikzlibrary{angles,quotes}
\usepackage{amsmath,amsthm,bm}
\newtheorem{theorem}{Theorem}
\ifCLASSINFOpdf
\else
\fi
\usepackage{amsmath}
\begin{document}
%
\title{Optimized Measures of Bipartite Quantum Correlation}
%
%
%

\author{Joshua~Levin,~Graeme Smith \\ \emph{JILA, University of Colorado, Boulder}}
\maketitle

\begin{abstract}
  How can we characterize different types of correlation between quantum systems? Since correlations cannot be generated locally, we take any real function of a multipartite state which cannot increase under local operations to measure a correlation.   Correlation measures that can be expressed as an optimization of a linear combination of entropies are particularly useful, since they can often be interpreted operationally. We systematically study such optimized linear entropic functions, and by enforcing monotonicity under local processing we identify four cones of correlation measures for bipartite quantum states.  This yields two new optimized measures of bipartite quantum correlation that are particularly simple, which have the additional property of being additive.  
\end{abstract}


%
\IEEEpeerreviewmaketitle

\section{Introduction}
%
%
%
%
\IEEEPARstart{Q}{uantifying} the correlations between disjoint subsystems of a quantum state is a fundamental problem in quantum information theory. Since correlations cannot be generated by local operations, measures of correlation must be non-increasing under local processing. For measures which are functions of the von Neumann entropy ($S(\rho) = -\text{Tr}~\rho\log\rho$), this is equivalent to being non-increasing under partial trace \footnote{This is because any processing can be written as an isometry followed by a partial trace, and the isometry will not affect entropies.}, i.e. 
\begin{align}
E(\rho_{(A_1A_2)(B_1B_2)})&\geq E(\rho_{A_1B_1})\label{mon}
\end{align}
for a correlation measure $E$.

In this work, we will take (\ref{mon}) to be the defining property of a bipartite correlation measure.  This property has been studied by \cite{alhejji2018monotonicity} for linear entropic quantities, but here we wish to identify bipartite correlation measures formed by minimizing a linear entropic quantity over all purifications of a state $\rho_{AB}$. More formally, we will study the space of quantities of the form
\begin{align}
E_\alpha(\rho_{AB}) = \inf_{\psi:\text{Tr}_{A'B'}\dyad{\psi}_{AA'BB'} = \rho_{AB}}f^\alpha(\dyad{\psi}_{AA'BB'})\nonumber
\end{align} 
where $\alpha\in\mathbb{R}^{15}$ and
\begin{align}
f^\alpha(\dyad{\psi}_{AA'BB'}) = \sum_{\emptyset\neq\mathcal{J}\subseteq\{A,B,A',B'\}}\alpha_{\mathcal{J}}S_{\mathcal{J}}\nonumber
\end{align}
(each entry of $\alpha$ corresponds to a non-empty subset of $\{A,B,A',B'\}$), and identify instances which satisfy (\ref{mon}).  Quantities of this form are of particular interest, since they often admit operational interpretations, usually in the form of bounds on performance in information theoretic tasks. Examples include the squashed entanglement \cite{doi:10.1063/1.1643788}, the entanglement assisted capacity \cite{bennett2002entanglement}, and the entanglement of purification \cite{doi:10.1063/1.1498001}.


In a pure state, the entropy of any subsystem is equal to the entropy of its complement, so we can remove redundancy from our search space by rewriting it as the set of quantities of the form
\begin{align}
E_\alpha(\rho_{AB}) = \inf_{\rho_{ABV}:\text{Tr}_V\rho_{ABV}=\rho_{AB}} f^\alpha(\rho_{ABV})\label{form2}
\end{align}
where $\alpha \in \mathbb{R}^7$ and
\begin{align}
f^\alpha(\rho_{ABV}) = \sum_{\emptyset\neq\mathcal{J}\subseteq\{A,B,V\}}\alpha_{\mathcal{J}}S_{\mathcal{J}}.
\end{align}
Note that the minimization is now over all extensions $\rho_{ABV}$, not only purifications $\dyad{\psi}_{AA'BB'}$.

By first examining the entanglement of purification, a well-known instance of (\ref{form2}) which satisfies (\ref{mon}), we are led to the construction of four convex polyhedral cones in $\mathbb{R}^7$.  These four cones consist of $\alpha$ vectors which give rise to optimized bipartite correlation measures which all satisfy (\ref{mon}).  We examine the extreme rays of these cones and find four nontrivial rays, two of which are new.
We study these two correlation measures and find several useful properties, including lower and upper bounds, additivity, and a relationship to the regularized entanglement of purification.

This paper is organized as follows.  In Section II we present a proof of the monotonicity of the entanglement of purification, in order to illustrate our method for identifying monotones of the form (\ref{form2}).  Guided by the proof in Section II, in Section III we define two different types of monotonicity and identify all monotones of each type.  In Section IV, we examine the monotones found in Section III and find that many are trivial in the sense that they are equal to $I_{A:B}$ or $0$.  After identifying some nontrivial monotones in Section IV, in Section V we go on to prove several important properties of these monotones. 
 
Throughout this paper, for compactness of notation, we will denote all entropic quantities using subscripts.  The entropy of subsystem $A$ will be denoted $S_A$, the entropy of $A$ conditioned on $B$ ($\equiv S_{AB} - S_B$) will be denoted $S_{A|B}$, the mutual information of $A$ and $B$ ($\equiv S_A + S_B - S_{AB}$) will be denoted $I_{A:B}$, and the mutual information of $A$ and $B$ conditioned on $V$ ($\equiv S_{AV} + S_{BV} - S_{ABV} - S_V$) will be denoted $I_{A:B|V}$.

\section{The entanglement of purification}

A well-known example of an optimized bipartite correlation measure is the entanglement of purification \cite{doi:10.1063/1.1498001}
\begin{align}
E_P(\rho_{AB}) = \inf_{\rho_{ABV}:\text{Tr}_V\rho_{ABV}=\rho_{AB}} S_{AV},\nonumber
\end{align}
i.e. $\alpha_{AV} = 1$ and $\alpha_\mathcal{J} = 0$ for all other $\mathcal{J}$. In this section we prove that (\ref{mon}) holds for $E_\alpha = E_P$, i.e. that $E_P$ is monotonically non-increasing under local processing of both subsystems of a bipartite state\footnote{This was first shown by \cite{doi:10.1063/1.1498001}, using a different method from the one used here.}. The proofs for monotocity under $A$-processing and $B$-processing are different, and point towards a method for identifying instances of (\ref{form2}) which satisfy (\ref{mon}).

First we show that $E_P(\rho_{AB})$ is monotone under $B$-processing.  For each extension $\rho_{AB_1B_2V}$ of $\rho_{AB_1B_2}$, we will construct an extension of $\rho_{AB_1}$ whose value of $S_{AV}$ is no greater than that of $\rho_{AB_1B_2V}$.  Given the extension $\rho_{AB_1B_2V}$ of $\rho_{AB_1B_2}$, consider the state
\begin{align}
\rho_{AB_1V}' = \text{Tr}_{B_2}\left[\rho_{AB_1B_2V}\right],\nonumber
\end{align}
which is an extension of the state $\rho_{AB_1}$.  Now note that 
\begin{align}
S_{AV}'\equiv S\left(\text{Tr}_{B_1}\left[\rho_{AB_1V}'\right]\right) = S\left(\text{Tr}_{B_1B_2}\left[\rho_{AB_1B_2V}\right]\right)\equiv S_{AV}.\nonumber
\end{align}
Therefore, every value of $S_{AV}$ achievable by an extension of the unprocessed state $\rho_{AB_1B_2}$ can also be achieved by an extension of the processed state $\rho_{AB_1}$. Thus, the minimum value of $S_{AV}$ for extensions of the processed state is no greater than the minimum value of $S_{AV}$ for extensions of the unprocessed state, which is exactly the statement that
\begin{align}
E_P(\rho_{AB_1B_2})\geq E_P(\rho_{AB_1}).\nonumber
\end{align}
This was a roundabout way of saying that $E_P(\rho_{AB})$ is monotone under $B$-processing because $S_{AV}$ itself (without the minimization) is monotone under $B$-processing.  

Now we show that $E_P(\rho_{AB})$ is monotone under $A$-processing.  Our method is the same, i.e., for each extension $\rho_{A_1A_2BV}$ of $\rho_{A_1A_2B}$, we construct an extension of $\rho_{A_1B}$ whose value of $S_{AV}$ is no greater than that of $\rho_{A_1A_2BV}$.  Given the extension $\rho_{A_1A_2BV}$ of $\rho_{A_1A_2B}$, consider the state $\rho'_{\hat{A}B\hat{V}}=\rho_{A_1B(A_2V)}'$, where $\hat{A} = A_1$ and $\hat{V} = A_2V$. This is the same global state but written as an extension of the processed state $\rho_{A_1B}$.  Now note that
\begin{align}
S_{\hat{A}\hat{V}}'\equiv S\left(\text{Tr}_B[\rho_{A_1B(A_2V)}']\right) = S\left(\text{Tr}_B\left[\rho_{A_1A_2BV}\right]\right)\equiv S_{AV},\nonumber
\end{align}
so we have shown that
\begin{align}
E_P(\rho_{A_1A_2B})\geq E_P(\rho_{A_1B}),\nonumber
\end{align}
which completes the proof of inequality (\ref{mon}).

\section{Monotones}

\subsection{Monotonicity types}

The main point to take away from the previous section is that for quantities of the form (\ref{form2}), there are two types of monotonicity we can identify.  One way for a quantity $E_\alpha$ to be monotonic under processing of a subsystem $X\in \{A,B\}$ is for the associated $f^\alpha$ to be monotonic under processing of $X$.  In this case, monotonicity of $E_\alpha$ is proved by starting with an extension of an unprocessed state and constructing from it an extension of an $X$-processed state by simply tracing out a subsystem $X_2$ of $X=X_1X_2$, as in the above proof of monotonicity of $E_P$ under $B$-processing.  The monotonicity of $f^\alpha$ then implies the monotonicity of $E_\alpha$.   This type of monotonicity (under, say, $A$-processing) is therefore characterized by the inequality
\begin{align}
f^\alpha(\rho_{A_1A_2BV})\geq f^\alpha(\rho_{A_1BV}).\label{ineq0}
\end{align}

But, as we saw for $E_P$, monotonicity of $f^\alpha$ is not necessary for monotonicity of $E_\alpha$.  All that is necessary is for $f^\alpha$ to be monotonic under some operation which constructs an extension of a processed state from an extension of an unprocessed state.  One such operation is a rearrangement of the subsystems making up the unprocessed state, as in the proof of monotonicity of the $E_P$ under $A$-processing.  In this case, monotonicity of $E_\alpha$ under $A$-processing is implied by monotonicity of $f^\alpha$ under the operation $\rho_{A_1A_2BV}\to\rho_{A_1B(A_2V)}$, i.e. placing $A_2$ with $V$ in order to write the state as an extension of the processed state $\rho_{A_1B}$.  This type of monotonicity, again under $A$-processing, is therefore characterized by the inequality
\begin{align}
f^\alpha(\rho_{A_1A_2BV})\geq f^\alpha(\rho_{A_1B(A_2V)}).\label{ineq1}
\end{align}
We will refer to monotonicity of the types characterized by (\ref{ineq0}) and (\ref{ineq1}) as 0-monotonicity and 1-monotonicity, respectively.  A quantity $E_\alpha$ can now be monotonic under both $A$- and $B$-processing in four (not mutually exclusive) ways that we can identify.  These quantities can be 00-, 01-, 10-, or 11-monotonic, where the first bit indicates whether $E_\alpha$ is 0- or 1-monotonic on $A$, and the second on $B$.  As an example, we have shown $E_P$ to be 10-monotonic.

Note that there is still a redundancy in the $\alpha$ vectors, due to a purification symmetry.  Given an extension $\rho_{ABV}$ of $\rho_{AB}$, we can form a canonical dual extension by purifying to $\rho_{ABVW}$, and tracing out $V$ to form $\rho_{ABW}$.  Now, given $f^\alpha$, there exists $f^\beta$ for which $f^\alpha(\rho_{ABV}) = f^\beta(\rho_{ABW})$ (implying $E_\alpha = E_\beta$). Using the fact that entropies of complimentary subsystems are equal in a pure state, we see that $\beta_{AV} = \alpha_{BV}$, $\beta_{BV} = \alpha_{AV}$, $\beta_{ABV} = \alpha_V$, and $\beta_V = \alpha_{ABV}$.  Also note that this symmetry takes 0-monotones to 1-monotones, and vice-verse. To see this, observe that under the purification symmetry, the operations defining 0- and 1-monotonicity ($\rho_{A_1A_2BV}\to\rho_{A_1BV}$ and $\rho_{A_1A_2BV}\to\rho_{A_1B(A_2V)}$, respectively) become
\begin{align}
&0:~~\rho_{A_1A_2BV}\to\rho_{A_1A_2BVW}\to\rho_{A_1BV(A_2W)}\to \rho_{A_1B(A_2W)}\nonumber
\\&1:~~\rho_{A_1A_2BV}\to\rho_{A_1A_2BVW}\to\rho_{A_1B(A_2V)W}\to\rho_{A_1BW}.\nonumber
\end{align}   
This means we need only study the 00- and 10-monotones, since the 11- and 01-monotones are redundant via the purification symmetry.

\subsection{Monotonicity cones}

Expanding (\ref{ineq0}) and (\ref{ineq1}) in terms of the coefficients $\alpha_{\mathcal{J}}$ and moving all terms to one side, we see that quantities $E_\alpha$ which are 0- or 1-monotonic on $A$ are those for which $\alpha$ satisfies
\begin{multline}
\alpha_AS_{A_2|A_1} + \alpha_{AB}S_{A_2|A_1B}
\\ + \alpha_{AV}S_{A_2|A_1V} + \alpha_{ABV}S_{A_2|A_1BV}\geq 0\stepcounter{equation}\tag{\arabic{equation}-0A}\label{0A}
\end{multline}
or
\begin{multline}
\alpha_AS_{A_2|A_1} + \alpha_{AB}S_{A_2|A_1B} 
\\- \alpha_{BV}S_{A_2|BV} - \alpha_VS_{A_2|V} \geq 0,\tag{\arabic{equation}-1A}\label{1A}
\end{multline}
respectively.  Here $S_{A|B} = S_{AB} - S_B$ is the conditional entropy. Swapping the roles of $A$ and $B$ in inequalities(\ref{0A}) and (\ref{1A}) gives inequalities 
\begin{multline}
\alpha_BS_{B_2|B_1} + \alpha_{AB}S_{B_2|B_1A} 
\\+ \alpha_{BV}S_{B_2|B_1V} + \alpha_{ABV}S_{B_2|B_1AV}\geq 0\tag{\arabic{equation}-0B}\label{0B}
\end{multline}
and
\begin{multline}
\alpha_BS_{B_2|B_1} + \alpha_{AB}S_{B_2|B_1A} 
\\- \alpha_{AV}S_{B_2|AV} - \alpha_VS_{B_2|V} \geq 0,\tag{\arabic{equation}-1B}\label{1B}
\end{multline}
satisfied by those $\alpha$ for which $E_\alpha$ is 0- or 1-monotonic on $B$, respectively.  The set of all $\alpha\in\mathbb{R}^7$ for which (\ref{0A}) or (\ref{1A}) is implied by strong subadditivity (SSA) ($I_{A:B|C}\geq 0$) \cite{doi:10.1063/1.1666274, pippenger2003inequalities} and weak monotonicity (WM) ($S_{C|A} + S_{C|B}\geq 0$) \cite{doi:10.1063/1.1666274} of the von Neumann entropy, for any 4-partite state $\rho_{A_1A_2BV}$, form convex polyhedral cones in $\mathbb{R}^7$.  Similarly, the set of all $\alpha\in\mathbb{R}^7$ for which (\ref{0B}) or (\ref{1B}) is implied by SSA and WM for any 4-partite state $\rho_{AB_1B_2V}$ also form convex polyhedral cones.  Since the intersection of two convex polyhedral cones is a convex polyhedral cone, the set of all 00-, 10-, 01-, and 11-monotonic quantities (i.e., those $\alpha$ which satisfy, respectively, (\ref{0A}) and (\ref{0B}),  (\ref{1A}) and (\ref{0B}), (\ref{0A}) and (\ref{1B}), (\ref{1A}) and (\ref{1B})) each form a convex polyhedral cone.  Using SAGE's\footnote{SAGE is a Python-based open-source mathematics software available at www.sagemath.org} rational convex polyhedral cone module, together with the constraints on entropy vectors implied by SSA, one can determine that the 00- and 10-cones are generated by the extreme rays given by the rows of Table \ref{tab:table1}.
\begin{table}[h]
\begin{center}
\caption{Rows are the extreme rays of the 00- and 10-monotone cones in $\mathbb{R}^7$.}
\label{tab:table1}
\begin{tabular}{c|c|c|c|c|c|c|c}
\text{Cone} & $\alpha_A$ & $\alpha_B$ & $\alpha_V$ & $\alpha_{AB}$ & $\alpha_{AV}$ & $\alpha_{BV}$ & $\alpha_{ABV}$
\\\hline
00 & 1 & 1 & 0 & -1 & 0 & 0 & 0
\\~& 1 & 0 & 0 & 0 & -1 & 0 & 0
\\~& 0 & 0 & 0 & 0 & 1 & 1 & -1
\\~& 0 & 0 & 0 & 1 & 0 & 0 & -1
\\~& 0 & 1 & 0 & 0 & 0 & -1 & 0
\\~& 0 & 0 & 1 & 0 & 0 & 0 & 0
\\~& 0 & 0 & -1 & 0 & 0 & 0 & 0
\\\hline
10 & 1 & 1 & 0 & -1 & 0 & 0 & 0
\\~& 0 & 0 & -1 & 0 & 0 & 1 & -1
\\~& 1 & 0 & -1 & 0 & 0 & 0 & 0
\\~& 0 & 1 & 0 & 0 & 0 & 0 & -1
\\~& 1 & 1 & 0 & 0 & 0 & -1 & 0
\\~& 0 & 0 & -1 & 1 & 0 & 0 & -1
\\~& 0 & 0 & 0 & 0 & 1 & 0 & 0
\\~& 0 & 0 & 0 & 0 & -1 & 0 & 0
\end{tabular}
\end{center}
\end{table}

\subsection{Non-negativity in $V$}

For certain $\alpha$, $E_\alpha$ is $-\infty$.  If, for some $\alpha$,
\begin{align}
\sum_{\substack{\mathcal{J}\subseteq\{A,B,V\}\\V\in\mathcal{J}}} \alpha_{\mathcal{J}}<0,\label{balance}
\end{align}
then we can achieve an arbitrarily large negative value of $f^\alpha$ for any state $\rho_{AB}$ by choosing an extension of the form $\rho_{ABV} = \rho_{AB}\otimes\mathbb{I}_k/k$, for sufficiently large $k$.  So for $\alpha$ satisfying (\ref{balance}), $E_\alpha$ is $-\infty$.  Therefore we are only interested in those $\alpha$ which satisfy
\begin{align}
\sum_{\substack{\mathcal{J}\subseteq\{A,B,V\}\\V\in\mathcal{J}}} \alpha_{\mathcal{J}}\geq 0,\nonumber
\end{align}
or equivalently, those $\alpha$ which satisfy
\begin{align}
\alpha\cdot (0,0,1,0,1,1,1)\geq 0.\label{balance2}
\end{align}
The set of all $\alpha$ satisfying (\ref{balance2}) form another convex cone $\mathcal{C}$ in $\mathbb{R}^7$, in fact they form the halfspace whose boundary is the plane through the origin with normal vector $(0,0,1,0,1,1,1)$.  Now we can intersect each of the three cones shown in Table \ref{tab:table1} with the cone $\mathcal{C}$, in order to keep only those $\alpha$ satisfying (\ref{balance2}).  The resulting cones (also obtained via SAGE's rational convex polyhedral cone module) are given by the extreme rays in Table \ref{tab:table2}.
\begin{table}[h]
\begin{center}
\caption{Rows are the extreme rays of the cones formed by intersecting the cones of 00- and 10-monotones with the cone $\mathcal{C}$ of vectors $\alpha\in\mathbb{R}^7$ which are non-negative in $V$ (i.e. satisfy (\ref{balance2})).}
\label{tab:table2}
\begin{tabular}{c|c|c|c|c|c|c|c||c}
\text{Cone} & $\alpha_A$ & $\alpha_B$ & $\alpha_V$ & $\alpha_{AB}$ & $\alpha_{AV}$ & $\alpha_{BV}$ & $\alpha_{ABV}$ & \text{label}
\\\hline
$\mathcal{C}\cap 00$ & 0 & 0 & 1 & 0 & 0 & 0 & 0 & 1
\\~& 1 & 1 & 0 & -1 & 0 & 0 & 0 & 2
\\~& 0 & 0 & -1 & 0 & 1 & 1 & -1 & 3
\\~& 0 & 0 & 1 & 1 & 0 & 0 & -1 & 4
\\~& 0 & 1 & 1 & 0 & 0 & -1 & 0 & 5
\\~& 1 & 0 & 1 & 0 & -1 & 0 & 0 & 6
\\\hline
$\mathcal{C}\cap 10$ & 0 & 0 & 0 & 0 & 1 & 0 & 0 & 7
\\~& 1 & 1 & 0 & -1 & 0 & 0 & 0 & 8
\\~& 0 & 0 & -1 & 0 & 1 & 1 & -1 & 9
\\~& 1 & 1 & 0 & 0 & 1 & -1 & 0 & 10
\\~& 0 & 0 & -1 & 1 & 2 & 0 & -1 & 11
\\~& 0 & 1 & 0 & 0 & 1 & 0 & -1 & 12
\\~& 1 & 0 & -1 & 0 & 1 & 0 & 0 & 13
\end{tabular}
\end{center}
\end{table}

\section{A closer look at the monotone cones}

The extreme rays of a convex polyhedral cone generate the cone via conical combinations (real linear combinations with non-negative coefficients), so any conical combination of the extreme rays of one of the two cones above gives a monotonic $E_\alpha$\footnote{Since the infimum of a conic combination is not generally equal to the conic combination of infima, the extreme rays are somewhat less priveleged in the optimized setting.  In other words there may be interesting quantities in these cones, other than the ones given by the extreme rays.  In this paper we do not discuss these quantities, but study the extreme rays as a starting point.}.  But in some cases these quantities can be trivial.  We will see that many of the rays in the cones in Table \ref{tab:table2} give $E_\alpha=0$, or $E_\alpha\propto I_{A:B}$.

We first examine $\mathcal{C}\cap 00$.  Extreme ray 1 gives $f^\alpha = S_V$, which is non-negative and achieves the value 0 for any $\rho_{AB}$ via the trivial extension.  So for this $\alpha$, we have $E_\alpha=0$.  Ray 2 gives $f^\alpha = I_{A:B}$, which gives $E_\alpha = I_{A:B}$ and the minimum is achieved by any extension.  Rays 4, 5, and 6 give, respectively, $f^\alpha = I_{AB:V}$, $f^\alpha = I_{B:V}$, and $f^\alpha = I_{A:V}$, which are non-negative and achieve the value 0 for any $\rho_{AB}$ via the trivial extension.  So for these three $\alpha$'s we also have $E_\alpha=0$.  So the only extreme ray of $\mathcal{C}\cap 00$ which is not minimized by the trivial extension and does not have $E_\alpha=0$ or $E_\alpha=I_{A:B}$ is ray 3, which gives $E_\alpha(\rho_{AB}) = E_{sq}(\rho_{AB})$ (the squashed entanglement \cite{doi:10.1063/1.1643788}, $E_{sq}(\rho_{AB}) = \inf_{\rho_{ABV}}(I_{A:B|V})$).  So any ray in $\mathcal{C}\cap 00$ which can be written as a conical combination of extreme rays that does not include the $E_{sq}$ ray will give $E_\alpha(\rho_{AB}) \propto I_{A:B}$ and is therefore trivial.  

$\mathcal{C}\cap 10$ is where we will find an abundance of nontrivial quantities.  There are only three trivial extreme rays, and they cannot be simultaneously minimized as in the two previous cones. Ray 8 is equal to rays 15 and 2, and again gives $E_\alpha = I_{A:B}$.  Ray 12 gives $f^\alpha = I_{B:AV}$, which SSA implies is bounded below by $I_{A:B}$, and achieves the value $I_{A:B}$ via the trivial extension.  Ray 13 gives $f^\alpha = S_A + S_{AV} - S_V$, which WM implies is bounded below by $I_{A:B}$, and achieves the value $I_{A:B}$ via any purification of $\rho_{AB}$. 

The four remaining extreme rays of $\mathcal{C}\cap 10$ are non-trivial.  Ray 9 is $E_{sq}$ which, interestingly, appears in all four monotonicity cones.  Rays 7, 10, and 11 are (up to a scaling by 1/2, the reason for which will be clear in the next section)
\begin{align}
f^{P} & = S_{AV}\nonumber
\\f^{Q} &= \frac{1}{2}(S_A + S_B + S_{AV} - S_{BV})\nonumber
\\f^{R} &= \frac{1}{2}(S_{AB} + 2S_{AV} - S_{ABV} - S_V),\nonumber
\end{align}
respectively. $f^P$ gives $E_P$, which we expected to find.  $E_Q$ and $E_R$ are new, and we will see that $E_Q$ and $E_R$ have several useful properties.
\section{Properties of $E_Q$ and $E_R$}

\subsection{Lower and upper bounds}

\begin{theorem}\footnote{This was proven for $E_P$ in \cite{doi:10.1063/1.1498001}}
$E_Q$ and $E_R$ satisfy 
\begin{align}
\frac{1}{2}I_{A:B}\leq E(\rho_{AB}) &\leq \min\{S_A,S_B\}.\label{bounds}
\end{align}
\end{theorem}
\begin{proof}
We start with the upper bound. Both $f^{Q}$ and $f^{R}$ achieve a value of $S_A$ via the trivial extension, and a value of $S_B$ via any purification.  Therefore $E_Q$ and $E_R$ satisfy the upper bound.  To prove the lower bound in (\ref{bounds}) for $E_Q$, observe that
\begin{multline}
2f^{Q} - I_{A:B} = S_{AV} + S_{AB} - S_{BV}
\\\geq S_{AV} + S_{AB} - S_B - S_V = S_{A|V} + S_{A|B}\geq 0,\nonumber
\end{multline}
where the first inequality follows from subadditivity and the second from WM.  Therefore $E_Q$ satisfies the lower bound.  To prove the lower bound for $E_R$, observe that
\begin{multline}
2f^{R} - I_{A:B} = (S_{AB} + S_{AV} - S_{ABV} - S_A) 
\\+ (S_{AB} + S_{AV} - S_V - S_B) 
\\ = I_{B:V|A} + S_{A|B} + S_{A|V}\geq I_{B:V|A}\geq 0,\nonumber
\end{multline}
where the first inequality follows from WM and the second from SSA.  Therefore $E_R$ satisfies the lower bound. 
\end{proof}

\subsection{Additivity}  

\begin{theorem}\label{Add}
$E_Q$ and $E_R$ are additive\footnote{$E_P$ is believed to be non-additive \cite{chen2012non}.}, i.e.
\begin{align}
E(\rho_{A_1B_1}\otimes\rho_{A_2B_2})= E(\rho_{A_1B_1}) + E(\rho_{A_2B_2}).\nonumber
\end{align}
\end{theorem}
\begin{proof} For brevity we will exclude the factor of 1/2. Fix two bipartite states $\rho_{A_1B_1}$ and $\rho_{A_2B_2}$, and form the bipartite state $\rho_{AB} = \rho_{A_1B_1}\otimes\rho_{A_2B_2}$ with $A=A_1A_2$ and $B=B_1B_2$. We start with $E_Q$. 
First we show that 
\begin{align}
E_Q(\rho_{A_1B_1}\otimes\rho_{A_2B_2})\geq E_Q(\rho_{A_1B_1}) + E_Q(\rho_{A_2B_2}).\nonumber
\end{align}
Let $\rho_{ABV}$ be an extension of $\rho_{AB}$.  Now let $V_1\equiv A_2V$ and $V_2\equiv B_1V$ and consider the extensions $\rho_{A_1B_1V_1}$ and $\rho_{A_2B_2V_2}$ of $\rho_{A_1B_1}$ and $\rho_{A_2B_2}$.  Since $\rho_{A} = \rho_{A_1}\otimes\rho_{A_2}$ and $\rho_{B} = \rho_{B_1}\otimes\rho_{B_2}$, we have
\begin{multline}
f^{Q}(\rho_{A_1B_1V_1}) + f^{Q}(\rho_{A_2B_2V_2})
\\= (S_{A_1} + S_{B_1} + S_{A_1V_1} - S_{B_1V_1}) 
\\+ (S_{A_2} + S_{B_2} + S_{A_2V_2} - S_{B_2V_2})
\\= S_A + S_B + S_{AV} - S_{B_1A_2V} + S_{B_1A_2V} - S_{BV}
\\=  S_A + S_B + S_{AV} - S_{BV} = f^Q(\rho_{ABV}).\nonumber
\end{multline}

Now we show that
\begin{align}
E_Q(\rho_{A_1B_1}) + E_Q(\rho_{A_2B_2})\geq E_Q(\rho_{A_1B_1}\otimes\rho_{A_2B_2}).\nonumber
\end{align}
Let $\rho_{A_1B_1V_1}$ and $\rho_{A_2B_2V_2}$ be extensions of $\rho_{A_1B_1}$ and $\rho_{A_2B_2}$, and consider the extension of $\rho_{AB}$ given by
\begin{align}
\rho_{ABV} = \rho_{A_1B_1V_1}\otimes \rho_{A_2B_2V_2},
\end{align}
with $V \equiv V_1V_2$. Now using the fact that $\rho_{AV} = \rho_{A_1V_1}\otimes\rho_{A_2V_2}$ and $\rho_{BV} = \rho_{B_1V_1}\otimes\rho_{B_2V_2}$ we have
\begin{multline}
f^Q(\rho_{ABV}) = S_A + S_B + S_{AV} - S_{BV}
\\=S_{A_1} + S_{A_2} + S_{B_1} + S_{B_2} + (S_{A_1V_1} + S_{A_2V_2}) - (S_{B_1V_1} + S_{B_2V_2}) 
\\= (S_{A_1} + S_{B_1} + S_{A_1V_1} - S_{B_1V_1}) + (S_{A_2} + S_{B_2} + S_{A_2V_2} - S_{B_2V_2})
\\= f^{Q}(\rho_{A_1B_1V_1})+ f^Q(\rho_{A_2B_2V_2}).\nonumber
\end{multline}
Therefore $E_Q$ is additive.

Now we wish to prove the same thing for $E_R$, i.e.
\begin{align}
E_R(\rho_{A_1B_1}\otimes\rho_{A_2B_2})= E_R(\rho_{A_1B_1}) + E_R(\rho_{A_2B_2}).\nonumber
\end{align}
First we show that 
\begin{align}
E_R(\rho_{A_1B_1}\otimes\rho_{A_2B_2})\geq E_R(\rho_{A_1B_1}) + E_R(\rho_{A_2B_2}).\nonumber
\end{align}
Let $\rho_{ABV}$ be an extension of $\rho_{AB}$.  Now let $V_1 \equiv A_2V$ and $V_2 \equiv A_1V$ (see \cite{PhysRevLett.118.040501}) and consider the extensions $\rho_{A_1B_1V_1}$ and $\rho_{A_2B_2V_2}$ of $\rho_{A_1B_1}$ and $\rho_{A_2B_2}$. We wish to show that $f^R(\rho_{ABV})\geq f^R(\rho_{A_1B_1V_1}) + f^R(\rho_{A_2B_2V_2}),$ i.e.
\begin{multline}
S_{A_1A_2B_1B_2} + 2S_{A_1A_2V} - S_{A_1A_2B_1B_2V} - S_{V}\\\geq(S_{A_1B_1} + 2S_{A_1VA_2} - S_{A_1B_1VA_2} - S_{VA_2})
\\+(S_{A_2B_2} + 2S_{A_2VA_1} - S_{A_2B_2VA_1} - S_{VA_1}).\label{stuff}
\end{multline}
Note that inequality (\ref{stuff}) is equivalent to
\begin{align}
I_{A_1:A_2|V} + I_{B_1:B_2|A_1A_2V} \geq I_{A_1B_1:A_2B_2}.\label{stuuff}
\end{align}
Since $I_{A_1B_1:A_2B_2}=0$ by assumption, (\ref{stuuff}) is true by SSA.

Now we show that
\begin{align}
E_R(\rho_{A_1B_1}) + E_R(\rho_{A_2B_2})\geq E_R(\rho_{A_1B_1}\otimes\rho_{A_2B_2}).\nonumber
\end{align}
Let $\rho_{A_1B_1V_1}$ and $\rho_{A_2B_2V_2}$ be extensions of $\rho_{A_1B_1}$ and $\rho_{A_2B_2}$, and as we did with $E_Q$, consider the extension of $\rho_{AB}$ given by
\begin{align}
\rho_{ABV} = \rho_{A_1B_1V_1}\otimes \rho_{A_2B_2V_2},
\end{align}
with $V\equiv V_1V_2$.
Now using the fact that systems 1 and 2 are in a product state, we have
\begin{multline}
f^R(\rho_{ABV}) = S_{AB} + 2S_{AV} - S_{ABV} - S_V
\\=S_{A_1B_1} + S_{A_2B_2} + 2(S_{A_1V_1} + S_{A_2V_2}) 
\\- (S_{A_1B_1V_1} + S_{A_2B_2V_2}) - (S_{V_1} + S_{V_2})
\\= S_{A_1B_1} + 2S_{A_1V_1} - S_{A_1B_1V_1} - S_{V_1}
\\+ S_{A_2B_2} + 2S_{A_2V_2} - S_{A_2B_2V_2} - S_{V_2}
\\= f^{R}(\rho_{A_1B_1V_1}) + f^R(\rho_{A_2B_2V_2}).\nonumber
\end{multline}
Therefore $E_R$ is also additive.
\end{proof}

\subsection{Relationship to regularized $E_P$}

The regularized $E_P$, defined as
\begin{align}
E_P^\infty(\rho_{AB}) = \lim_{n\to\infty}\frac{1}{n}E_P(\rho_{AB}^{\otimes n}),
\end{align}
has an important operational interpretation.  $E_P^\infty$ is the number of EPR pairs required to create $\rho_{AB}$ using only local operations and asymptotically vanishing communication \cite{doi:10.1063/1.1498001}.  In general, $E_P^\infty$ is difficult to calculate.  But there is a relationship between $E_P^\infty$ and the quantities $E_Q$ and $E_R$, which may provide a way to learn about $E_P^\infty$.

\begin{theorem}
\begin{align}
E(\rho_{AB})\leq E_P^\infty(\rho_{AB}),\nonumber
\end{align}
for $E=E_Q,E_R$.
\end{theorem}

\begin{proof}
First note that 
\begin{align}
f^{Q} - f^{P} &= \frac{1}{2}(S_A + S_B - S_{AV} - S_{BV}) \nonumber
\\&= -\frac{1}{2}(S_{V|A} + S_{V|B})\nonumber
\\&\leq 0\nonumber
\end{align}
and
\begin{align}
f^{R} - f^{P} &= \frac{1}{2}(S_{AB} - S_V - S_{ABV})\nonumber
\\&= -\frac{1}{2}(S_V + S_{V|AB})\nonumber
\\&\leq 0,\nonumber
\end{align}
where both inequalities follow from WM.  Therefore $E_P$ is lower bounded by both $E_Q$ and $E_R$.  Additivity of $E_Q$ and $E_R$ (Thm. \ref{Add}) now allows us to write
\begin{align}
E(\rho_{AB}) = \frac{1}{n}E(\rho_{AB}^{\otimes n})\leq\frac{1}{n}E_P(\rho_{AB}^{\otimes n})~~\forall n,\nonumber
\end{align}
for $E = E_Q,E_R$. Taking the $n\to\infty$ limit gives the theorem.
\end{proof}

\section{Evaluation of $E_P$, $E_Q$ and $E_R$}

\subsection{Pure state}

For any pure state $\ket\psi_{AB}$, all extensions are of the form $\rho_{ABV} = \dyad{\psi}_{AB}\otimes\rho_V$, which makes calculation of $E_P$, $E_Q$, and $E_R$ trivial:
\begin{align}
E_P(\ket\psi_{AB}) &= \inf_{\rho_{ABV}}S_{AV} = \inf_{\rho_{ABV}}(S_A + S_V) = S_A\nonumber
\\E_Q(\ket\psi_{AB})&= \frac{1}{2}\inf_{\rho_{ABV}}(S_A + S_B + S_{AV} - S_{BV})\nonumber
\\&= S_A + \frac{1}{2}\inf_{\rho_{ABV}}(S_A + S_V - S_B - S_V)\nonumber
\\&=S_A\nonumber
\\E_R(\ket\psi_{AB})&=\frac{1}{2}\inf_{\rho_{ABV}}(S_{AB} + 2S_{AV} - S_{ABV} - S_V)\nonumber
\\&=\frac{1}{2}\inf_{\rho_{ABV}}(2(S_A + S_V) - (S_{AB} + S_V) - S_V)\nonumber
\\&=S_A.\nonumber
\end{align}
So for a pure state, $E_P = E_Q = E_R = S_A = S_B$.

\subsection{Classically correlated state}

We can also evaluate all three correlation measures for the classically correlated state
\begin{align}
\rho_{AB}^c = \sum_ip_i\dyad{ii}_{AB}.\nonumber
\end{align}
First note that an arbitrary extension of $\rho_{AB}^c$ takes the form 
\begin{align}
\rho_{ABV}^c=\sum_{i,j}\sqrt{p_ip_j}\dyad{ii}{jj}\otimes\rho_V^{ij},\nonumber
\end{align}
with $\text{Tr}~\rho_V^{ij} = \delta_{ij}$.
From this we can see that $\rho_{AV}^c = \sum_ip_i\dyad i_A\otimes\rho_V^{ii}$, and $\rho_{BV}^c = \sum_ip_i\dyad i_B\otimes\rho_V^{ii}$. 
\subsubsection{$E_P$} Since $S_{V|A} = \Sigma_ip_iS(\rho_V^{ii})\geq 0,$ we have that $f^P = S_A + S_{V|A}\geq S_A$, which is saturated by the trivial extension.  Therefore $E_P(\rho_{AB}^c) = S_A$.

\subsubsection{$E_Q$} From the form of $\rho_{AV}^c$ and $\rho_{BV}^c$, we see that $S_{AV} = S_{BV}$ for any extension of $\rho_{AB}^c$, so $f^{Q} = \frac{1}{2}(S_A + S_B + S_{AV} - S_{BV}) = S_A$, so $E_Q(\rho_{AB}^c) = S_A = H(\{p_i\})$.

\subsubsection{$E_R$} By (\ref{bounds}), $E_R\geq \frac{1}{2}I_{A:B} = \frac{1}{2}H(\{p_i\})$ for the state $\rho_{AB}^c$.  It is easy to check that this value is achieved by the extension
\begin{align}
\rho_{ABV}^c = \sum_ip_i\dyad{iii}_{ABV},\nonumber
\end{align}
so $E_R(\rho_{AB}^c) = \frac{1}{2}I_{A:B} = \frac{1}{2}H(\{p_i\})$.

\subsection{Symmetric or antisymmetric state}

For states with support entirely within the symmetric or antisymmetric subspace, we have

\subsubsection{$E_P$} In this case, \cite{christandl2005uncertainty} showed that $E_P = S_A$  and that $E_P$ is additive.
\subsubsection{$E_Q$} States $\rho_{ABV}$ with the reduced state on $\rho_{AB}$ supported within the symmetric or antisymmetric subspace are invariant under the swap operator $F_{AB} = \sum_{ij}\dyad{ij}{ji}_{AB}$, i.e.
\begin{align}
\rho_{BAV} = (F_{AB}\otimes\mathbb{I}_V)\rho_{ABV}(F_{AB}^\dagger\otimes\mathbb{I}_V) = \rho_{ABV}.\nonumber
\end{align}

To see this, note that the most general pure state $\ket\psi_{ABV}$ for which $\text{Tr}_V\dyad\psi_{ABV}$ is supported entirely in the symmetric or antisymmetric subspace of $\mathcal{H}_A\otimes\mathcal{H}_B$ is of the form
\begin{align}
\ket\psi_{ABV} = \sum_i\ket{\phi_i}_{AB}\ket{\xi_i}_C,\label{symm}
\end{align}
where all $\ket{\phi_i}_{AB}$ are symmetric or all $\ket{\phi_i}_{AB}$ are antisymmetric.  In the former case we have 
\begin{align}
F_{AB}\ket\psi_{ABV} = \sum_i\left(F_{AB}\ket{\phi_i}_{AB}\right)\ket{\xi_i}_C = \ket\psi_{ABV},\nonumber
\end{align}
while in the latter case we have
\begin{align}
F_{AB}\ket\psi_{ABV} = \sum_i\left(F_{AB}\ket{\phi_i}_{AB}\right)\ket{\xi_i}_C = -\ket\psi_{ABV}.\nonumber
\end{align}
In both cases, we see that 
\begin{align}
F_{AB}\dyad\psi_{ABV}F_{AB}^\dagger = \dyad\psi_{ABV}.\nonumber
\end{align}
Since $\rho_{ABV}$ with symmetric or antisymmetric $\rho_{AB}$ is generally a mixture of states of the form (\ref{symm}), invariance under $F_{AB}$ follows.  In particular, this means that for any extension $\rho_{ABV}$ of a symmetric or antisymmetric state $\rho_{AB}$, we have $S_{AV} = S_{BV}$.  Therefore $f^{Q} = S_A$, so $E_Q(\rho_{AB}) = S_A$.

\subsubsection{$E_R$} Again, for any extension $\rho_{ABV}$, we have $S_{AV} = S_{BV}$, so
\begin{align}
f^{R} &= \frac{1}{2}(S_{AB} + (S_{AV} + S_{BV}) - S_V - S_{ABV})\nonumber
\\&= \frac{1}{2}(S_{AB} + I_{A:B|V}).\nonumber
\end{align}
Therefore $E_R(\rho_{AB}) = \frac{1}{2}S_{AB} + E_{sq}(\rho_{AB})$.

\section{Conclusion}

We have identified four quantities of the form (\ref{form2}) which have the properties of monotonicity and the lower and upper bounds in (\ref{bounds}).  Two of them, the entanglement of purification and the squashed entanglement, have appeared in the literature \cite{doi:10.1063/1.1498001}\cite{doi:10.1063/1.1666274}\cite{doi:10.1063/1.1643788} and have been thoroughly studied.  The other two, which we have called $E_Q$ and $E_R$, are new.  We have shown that these two quantities are additive. We have also shown that they are lower bounds for $E_P^\infty$, which could potentially provide a calculational handle for $E_P^\infty$.  

It should be noted that the search method used in this paper is not exhaustive.  In particular, we are restricted to showing monotonicity by (\ref{ineq0}) and (\ref{ineq1}), while there may be other inequalities which imply monotonicity of a optimized correlation measure. A possible future research thrust is identifying such new correlation measures, or confirming their non-existence. Additionally, we have not yet found operational interpretations for the correlation measures $E_Q$ and $E_R$ identified in this paper.  The examples presented in Section VI suggest that while $E_Q$ and $E_R$ capture both classical and quantum correlation, $E_R$ distinguishes between the two, whereas $E_Q$ seems not to tell them apart. Finally, we have considered only bipartite correlations in this work, but we expect that a search for optimized multipartite correlation measures may yield new and exciting formulas to study and interpret. Understanding the two new correlation measures identified here, as well as other possibly new optimized correlation measures (both bipartite and multipartite) will be a step forward towards the broader goal of understanding the structures of quantum correlations.


%



\section*{Acknowledgment}

We would like to thank Mohammad Alhejji, Michael Perlin, and the anonymous referees for useful conversations and insightful comments.  This work was supported by NSF CAREER award CCF 1652560.

\ifCLASSOPTIONcaptionsoff
  \newpage
\fi



%
\bibliography{\jobname}
\bibliographystyle{ieeetr}

%


\begin{IEEEbiographynophoto}{Graeme Smith}
Graeme Smith has been an assistant professor of physics at the University of Colorado Boulder, and an associate fellow at JILA since 2016.  He has recently been awarded an NSF CAREER award.  Graeme was previously a research staff member (2010-2016) and postdoc (2007-2010) at  IBM’s TJ Watson Research Center.  He was a research associate in computer science at the University of Bristol from 2006-2007, and received a PhD in theoretical physics from Caltech in 2006.  Graeme is interested in computation and communication in noisy settings, quantum Shannon theory, network information theory, coding theory, cryptography, quantum estimation and detection, optical communications, and the physics of information.
\end{IEEEbiographynophoto}

\begin{IEEEbiographynophoto}{Joshua Levin}
Joshua Levin has been a graduate student in the department of physics at University of Colorado Boulder since 2016.  He received a BA in mathematics from the University of Michigan Ann Arbor in 2011, and an MS in Physics from the University of Colorado Boulder in 2019.  Joshua is interested in the implications of quantum information and correlation theory in quantum gravity, as well as network information theory and quantum Shannon theory. 
\end{IEEEbiographynophoto}






\end{document}